\documentclass{amsart}

\usepackage{amsmath,amsthm,amssymb,verbatim, graphicx, url}
\usepackage{bbm, stmaryrd}
\usepackage{extarrows}
\usepackage{tikz-cd}

\newtheorem{theorem}[subsection]{Theorem}
\newtheorem{lemma}[subsection]{Lemma}
\newtheorem{corollary}[subsection]{Corollary}
\newtheorem{proposition}[subsection]{Proposition}

\theoremstyle{definition}
\newtheorem{definition}[subsection]{Definition}
\newtheorem{remark}[subsection]{Remark}

\DeclareMathOperator{\A}{\mathbb{A}\mathbb{A}}
\DeclareMathOperator{\N}{\mathbb{N}}
\DeclareMathOperator{\R}{\mathbb{R}}
\DeclareMathOperator{\Z}{\mathbb{Z}}
\renewcommand{\aa}{\mathbf{a}}
\DeclareMathOperator{\bb}{\mathbf{b}}

\DeclareMathOperator{\mm}{\mathbf{m}}
\DeclareMathOperator{\uu}{\mathbf{u}}
\DeclareMathOperator{\vv}{\mathbf{v}}
\DeclareMathOperator{\ww}{\mathbf{w}}
\DeclareMathOperator{\xx}{\mathbf{x}}
\DeclareMathOperator{\yy}{\mathbf{y}}
\DeclareMathOperator{\zz}{\mathbf{z}}
\DeclareMathOperator{\calA}{\mathcal{A}}
\DeclareMathOperator{\calD}{\mathcal{D}}
\DeclareMathOperator{\calF}{\mathcal{F}}
\DeclareMathOperator{\calG}{\mathcal{G}}
\DeclareMathOperator{\calH}{\mathcal{H}}
\DeclareMathOperator{\calI}{\mathcal{I}}
\DeclareMathOperator{\calJ}{\mathcal{J}}
\DeclareMathOperator{\calM}{\mathcal{M}}
\DeclareMathOperator{\calR}{\mathcal{R}}
\DeclareMathOperator{\calX}{\mathcal{X}}
\DeclareMathOperator{\zero}{\mathbf{0}}

\DeclareMathOperator{\im}{\textup{Im}}
\DeclareMathOperator{\Span}{\textup{Span}}
\DeclareMathOperator{\Spec}{\textup{Spec}}
\DeclareMathOperator{\lex}{\ell}
\DeclareMathOperator{\one}{\mathbbm{1}}
\DeclareMathOperator{\HOT}{\textup{higher order terms}}
\DeclareMathOperator{\Sch}{\mathfrak{Sch}}
\DeclareMathOperator{\colim}{\textup{colim}}

\DeclareMathOperator{\into}{\hookrightarrow}

\title{Numeric Invariants from Multidimensional Persistence \\ \emph{May 2015 draft}}

\author{Jacek Skryzalin and Gunnar Carlsson}
\address{Department of Mathematics, Stanford University}
\email{jskryzal@math.stanford.edu}
\email{gunnar@math.stanford.edu}

\begin{document}

\begin{abstract}
We extend the results of Adcock, Carlsson, and Carlsson (\cite{algebraic_functions}) by constructing numeric invariants from the computation of a multidimensional persistence module as given by Carlsson, Singh, and Zomorodian in \cite{computing_multid_persistence}.
\end{abstract}

\maketitle

The use of topology to study point cloud data has been well established (\cite{top_and_data}, \cite{acta_paper}). Given a finite metric space (e.g., a finite subset of $\R^n$), one first constructs a filtered complex which attempts to approximate the shape of the underlying data. Commonly used complexes include the Vietoris-Rips complex, the Cech complex, the $\alpha$-complex, and the witness complex. The \emph{persistent homology} of a filtered complex is an abstract algebraic entity which combines information about the homology of the levelwise complexes of a filtered simplicial complex and the maps on homology induced by the filtration maps of the complex. The process of constructing a filtered complex and taking persistent homology provides abstract algebraic information about the original point cloud data.

It is difficult to interpret raw calculations of persistent homology from a geometric and intuitive standpoint. This is partially remedied by Adcock, Carlsson, and Carlsson, who have successfully studied ways of interpreting persistent homology geometrically through the construction of numeric invariants \cite{algebraic_functions}. They produce an infinite family of functions, each of which takes as input any one-dimensional persistence module (the most notable of such objects being the persistent homology of a one-dimensional filtered complex) and outputs a nonnegative number which has a concrete interpretation in terms of the geometry of the filtered complex. If the filtered complex was obtained from point cloud data, these values provide information about the size and density of prominent geometric features of the point cloud. More importantly, these values can then be used as features in machine learning algorithms. 

Unfortunately, the method of \cite{algebraic_functions} does not generalize nicely to multidimensional persistence modules. The construction of the algebraic functions in \cite{algebraic_functions} relies heavily on the classification theorem of finitely generated modules over a PID and the categorical equivalence between one-dimensional persistence modules and finitely generated modules over a PID. Multidimensional persistence modules are categorically equivalent to finitely presented graded $\R[x_1, ..., x_n]$-modules, for which there is no analogous classification theorem. Furthermore, the results of \cite{theory_multid_persistence} show that there is no complete discrete invariant for multidimensional persistence; any numeric invariants that we provide must necessarily be incomplete.

It is our goal to provide functions which generalize those of \cite{algebraic_functions}. Although our functions are identical to those of \cite{algebraic_functions} in the case of one-dimensional persistence, our construction is measure-theoretic rather than algebraic. As such, our functions can be defined just as easily on the space of multidimensional persistence modules as on the space of one-dimensional persistence modules.

In Section \ref{preliminaries} we review multidimensional persistence and introduce convenient notation. In Section \ref{the-ring} we calculate the ring of $K$-finite functions on a convenient set of multidimensional persistence modules, which we then extend to all multidimensional persistence modules in Section \ref{calculate-invariants}. In Section \ref{group_completions_geo_insights}, we present results pertaining to the power of the invariants discussed in Sections \ref{the-ring} and \ref{calculate-invariants}, and we also propose an alternative way to calculate these invariants. Finally, in Section \ref{recovering}, we show how to recover information about a multidimensional persistence module from the functions defined in Section \ref{calculate-invariants}.

\section{Preliminaries and Conventions}\label{preliminaries}

\subsection*{Notation:} In this paper, $k$ will denote a field of arbitrary characteristic. We will denote elements of $k^n$ by bold letters (i.e., $\aa \in k^n$), but we will denote their components by italic letters (i.e., $a_j \in k$). For $\aa, \bb \in k^n$ and $\mm \in \N^n$, we will denote the sum $\sum_{j = 1}^n a_i + b_i$ by $\aa + \bb$ and the product $\prod_{i=1}^n a_i^{m_i}$ by $\aa^{\mm}$. If additionally $c,d \in k$, we denote the quantity $\left(\sum_{i=1}^n a_i\right) + c$ by $\aa + c$, the quantity $\prod_{i=1}^n c^{m_i}$ by $c^{\mm}$, and the quantity $\prod_{i=1}^n (a_i + c)(b_i + d)$ by $(\aa + c)(\bb + d)$. Note that we will frequently make use of these abbreviations for \emph{constants} in $k^n$, but we will never make use of these abbreviations for \emph{variables}. That is, for variables $\xx, \yy \in k^n$, the quantity $\xx + \yy$ denotes the vector sum of $\xx$ and $\yy$ in $k^n$.

For a variable \emph{or} constant $\xx \in \R^n$, denote by $\lfloor \xx \rfloor$ (respectively, $\lceil \xx \rceil$) the elements of $\R^n$ obtained by taking the floor (respectively, ceiling) of each of the components $x_j$ of $\xx$. Similarly, for $\xx, \yy \in \R^n$, we say that $\xx \leq \yy$ if $x_j \leq y_j$ for all $i$. Additionally, we will use $\leq_{\lex}$ to denote the lexicographic order on $\R^n$.

Finally, unless otherwise noted, sets in this paper will be sets with repetition.

\subsection*{Multidimensional persistence:}

We begin by reviewing the concept of persistence.

\begin{definition}
A persistence module $M$ indexed by the partially ordered set $V$ is a family of $k$-modules $\{M_{\vv}\}_{\vv \in V}$ together with homomorphisms $\phi_{\uu, \vv} : M_{\uu} \to M_{\vv}$ for all $\uu \leq \vv$, such that $\phi_{\uu, \vv} \circ \phi_{\vv, \ww} = \phi_{\uu, \ww}$ whenever $\uu \leq \vv \leq \ww$.
\end{definition}

In this paper, the indexing set $V$ with either be $\N^n$ or $\R^n$ for some $n \in \N_{> 0}$. When the indexing set $V$ is $\N^1$, we refer to $M$ as a one-dimensional persistence module. When $V$ is $\N^n$ for $n > 1$, we refer to $M$ as a multidimensional persistence module. Furthermore, when $V$ is $\N^n$, we refer to $M$ as an integral persistence module; when $V$ is $\R^n$, we refer to $M$ as a real persistence module.

\begin{definition}
Given a persistence module $M$ indexed over $\N^n$, we can define an $n$-graded module $\alpha(M)$ over $k[x_1, ..., x_n]$ via the following: \[\alpha(M) = \bigoplus_{\vv \in \N^n} M_{\vv},\] where the $k[x_1, ..., x_n]$-module structure is given by \[\xx^{\vv - \uu}m_{\uu} = \phi_{\uu, \vv}(m_{\uu})\] for $m_{\uu} \in M_{\uu}$ whenever $\uu \leq \vv$. 
\end{definition}

Furthermore, we have the following theorem:

\begin{theorem}[\cite{theory_multid_persistence}]\label{correspond}
The correspondence $\alpha$ defines an equivalence of categories between the category of finite persistence modules over $k$ and the category of finitely presented $n$-graded modules over $k[x_1, ..., x_n]$.
\end{theorem}

Theorem \ref{correspond} allows us to interpret persistence modules indexed by $\N^n$ as finitely presented graded $k[x_1, ..., x_n]$-modules. This correspondence allows us to study persistence modules using the well-developed theory of graded modules. 

\begin{definition}
Define $\calM(k,n)$ to be the category of all finite persistence $k$-modules indexed by $\N^n$. For an element $M \in \calM(k,n)$, we will often consider $M$ simultaneously as a persistence module and as a finitely presented $n$-graded $k[x_1, ..., x_n]$-module.
\end{definition}

Furthermore, it should be noted that we will usually focus on the case $k = \R$.

\subsection*{Algebraic geometry:}

In general, for a scheme $X$, we denote the ring of global sections of $X$ by $A[X]$. Let $\A_n^m$ denote the affine $2mn$ $\R$-space \[\A_n^m = \Spec\left(\R[x_{ij},\, y_{ij}]_{1 \leq i \leq m,\, 1 \leq j \leq n}\right).\] We consider $\A_n^m$ as having coordinates $(\xx_1, \yy_1, ..., \xx_m, \yy_m)$. The symmetric group $S_m$ acts on $A[\A_n^m]$ by simultaneously permuting the $\xx_i$ and $\yy_i$. Denote by $A[\A_n^m]^{S_m}$ the elements of $A[\A_n^m]$ invariant under this $S_m$-action.

\section{The Ring of $K$-finite Algebraic Functions on Persistent Cubes}\label{the-ring}

In this section, we calculate the ring of $K$-finite algebraic functions on a simple and easily studied subset $\calR(k,n) \subseteq \calM(k,n)$. The definition of $\calR(k,n)$, given below, is very easily parameterized, and hence is easily analyzed from the standpoint of algebra and algebraic geometry.  This section extends a result of \cite{algebraic_functions}.

\subsection{Defining $\widetilde{Sp}$}

Let $\calR^m(k,n) \subseteq \calM(k,n)$ consist of all multidimensional persistence modules isomorphic to those of the form \[\bigoplus_{i=1}^{m^\prime} \left(\frac{k[x_1, ..., x_n]}{\left(x_1^{d_{i1}},\cdots, x_n^{d_{in}}\right)}\right)_{\vv_i},\] where $m^{\prime} \leq m$, $\vv_i$ represents a grading and $0 < d_{ij} < \infty$. Define $\calR(k,n) \subseteq \calM(k,n)$ by \[\calR(k,n) = \bigcup_{m=1}^\infty \calR^m(k,n).\] Each summand of an element of $\calR(k,n)$ can be represented by an element in $\N^{2n}$: $(v_{i1}, ..., v_{in}, v_{i1} + d_{i1}, ..., v_{in} + d_{in})$. The first $n$ coordinates can be viewed as when the summand is ``born", and the last $n$ coordinates represent when the summand ``dies." Note additionally that the ordering of each summand within the direct sum decomposition given above is irrelevant. We now formulate this intuition algebraically.

\begin{definition}
Let $\calJ$ denote any object in a category. Define $Sp^m(\calJ)$ to be the colimit, if it exists, of the diagram
\begin{center} \includegraphics{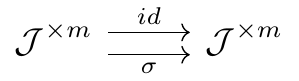}, \end{center}
where $\sigma$ varies over all elements in the symmetric group $S_m$.
\end{definition}

We have natural maps \[+ : Sp^m(\calJ) \times Sp^{m^\prime}(\calJ) \to Sp^{m+m^\prime}(\calJ).\] Fixing some basepoint $j_0 : * \to \calJ$, we have natural inclusions $\iota_m : Sp^{m}(\calJ) \into Sp^{m+1}(\calJ)$, where $\iota_m$ is defined as the composite \[Sp^{m}(\calJ) \longrightarrow Sp^{m}(\calJ) \times * \xlongrightarrow{id \otimes j_0} Sp^{m}(\calJ) \times Sp^{1}(\calJ) \xlongrightarrow{+} Sp^{m+1}(\calJ).\] Define \[Sp^\infty(\calJ) = \varinjlim_{m} Sp^m(\calJ).\] Furthermore, we have canonical inclusions $Sp^{m}(\calJ) \into Sp^{\infty}(\calJ)$. Moreover, $Sp^{\infty}(\calJ)$ is a commutative monoid generated by $Sp^1(\calJ)$ with monoid operation given by $+$.

If $\calJ$ is a scheme, then the inclusions $\iota_m$ induce maps \[\iota_m^* : A[Sp^{m+1}(\calJ)] \to A[Sp^{m}(\calJ)].\] In this case, we also have \[A[Sp^\infty (\calJ)] = \varprojlim_m A[Sp^m (\calJ)].\] We fix as the basepoint of $Sp^1(\A_n^1)$ the map \[j_0^* : * \rightarrow \A_n^1 = Sp^1(\A_n^1)\] defined as the dual of the ``evaluation at $0$" map \[j_0 : \R[\xx, \yy] \rightarrow \R \qquad \textup{defined by} \qquad x_i \mapsto 0, \quad y_i \mapsto 0.\] 

Letting $J(n) = \{ (\xx, \yy) \in \Z^n \times \Z^n \mid \xx < \yy \} \cup \{(\zero,\zero)\}$, we see that our former intuition about $\calR^m(k,n)$ translates into the following statement.

\begin{lemma}
There is a set-isomorphism between $\calR^m(k,n)$ and $Sp^m(J(n))$:
\[
\bigoplus_{i=1}^m \left(\frac{k[x_1, ..., x_n]}{\left(x_1^{d_{i1}},\cdots, x_n^{d_{in}}\right)}\right)_{\vv_i} \xlongrightarrow{\cong} \left\{(v_{i1}, ..., v_{in}, v_{i1} + d_{i1}, ..., v_{in} + d_{in}) \right\}_{1 \leq i \leq m}.
\]
This set-isomorphism extends to a set-isomorphism between $\calR(k,n)$ and $Sp^{\infty}(J(n))$.
\end{lemma}

We now discuss how we might approach our study of $Sp^m(J(n))$ from an algebraic geometric viewpoint.

\begin{lemma} \label{sp_lemma}
There is a set-isomorphism between $Sp^m(J(n))$ and a well-chosen subset of the $\R$-points of $Sp^m(\A_n^1)$. This set-isomorphism extends to a set-isomorphism between $Sp^{\infty}(J(n))$ and a well-chosen subset of the $\R$-points of $Sp^{\infty}(\A_n^1)$.
\end{lemma}
\begin{proof}
By symmetrizing the identification of $\R^n \times \R^n$ with the set of $\R$-points of $\A_n^1$, we may identify $Sp^m(\R^n \times \R^n)$ with a subset of the $\R$-points of $Sp^m(\A_n^1)$. In particular, we identify $Sp^m(\R^n \times \R^n)$ with the subset of $\R$-points of $Sp^m(\A_n^1)$ which factor through $\Spec(\R[\xx_i, \yy_i])$ (where the map $\Spec(\R[\xx_i, \yy_i]) \to Sp^m(\A_n^1)$ is induced by the canonical inclusion of the ring of multi-symmetric polynomials into its ambient polynomial ring). Hence, \[Sp^m(J(n)) \subseteq Sp^m(\R^n \times \R^n) \subseteq \left(\R\textup{-points of }Sp^m(\A_n^1)\right).\]
\end{proof}

\begin{remark} \label{sp_remark}
We emphasize that $Sp^m(\R^n \times \R^n)$ is identified with a \emph{subset} of the $\R$-points of $Sp^m(\A_n^1)$. The fact that we work over the non-algebraically closed field $\R$ is crucial here. For example, in the case $m=2$ and $n=1$, the $\R$-point of $Sp^2(\A_1^1)$ induced by the homomorphism $A[Sp^2(\A_1^1)] \to \R$ defined by \[y_1 \mapsto 0 \qquad \quad y_2 \mapsto 0 \qquad \quad x_1 + x_2 \mapsto 0 \qquad \quad x_1^2 + x_2^2 \mapsto -1\] is not identified with any element of $Sp^m(\R^n \times \R^n)$.
\end{remark}

The previous two lemmas combine to yield the following:

\begin{corollary} \label{rect_to_sp}
There is a set-isomorphism between $\calR(k,n)$ and the subset of the $\R$-points of $Sp^{\infty}(\A_n^1)$ induced by homomorphisms \[\varphi : \R[\xx_i, \yy_i]_{1 \leq i \leq m} \to \R\] such that $(\varphi(\xx_i), \varphi(\yy_i)) \in J(n)$ for all $i$.
\end{corollary}

\begin{remark}
In a real (rather than integral) formulation of multidimensional persistence, one would instead define $J(n) = \{ (\xx, \yy) \in \R^n \times \R^n \mid \xx < \yy \} \cup \{(\zero, \zero)\}$. In either case, we view $J(n)$ as a subset of $\R^n \times \R^n$; all results discussed above are thus valid in either setting.
\end{remark}

In the definition of $\calR(k,n)$ given above, we required that the $d_{ij}$ be strictly positive. That is, we require strict inequality in our definition of $J(n)$. This inequality is lost if we work with $\R^n \times \R^n$ rather than with $J(n)$. Nevertheless, we still wish to encode algebraically that we wish to disregard any summand of any element of $Sp^{\infty}(\R^n \times \R^n)$ of the form $(\xx, \yy)$ where one coordinate $x_i$ of $\xx$ is equal to the corresponding coordinate $y_i$ of $\yy$. Put differently, we intuitively think of $(\xx, \yy)$ as the opposite vertices of a cube, and we wish to disregard any cubes of volume $0$.

To this end, define \[\widetilde{Sp}(\R^n \times \R^n) = \frac{\coprod_{m} Sp^m(\R^n \times \R^n)}{\simeq},\] where $\simeq$ is the equivalence relation generated by all relations of the form \[\{(\xx_1, \yy_1), (\xx_2, \yy_2), ..., (\xx_m, \yy_m), (\zz, \zz^\prime) \} \simeq \{(\xx_1, \yy_1), (\xx_2, \yy_2), ..., (\xx_m, \yy_m)\},\] where one of the coordinates of $\zz$ equals one of the coordinates of $\zz^\prime$.

We now generalize our definition of $\widetilde{Sp}(\R^n \times \R^n)$ to the algebraic geometric setting to define $\widetilde{Sp}(\A_n^1)$. We define $\widetilde{Sp}(\A_n^1)$ to be the colimit of the diagram
\begin{center} \includegraphics{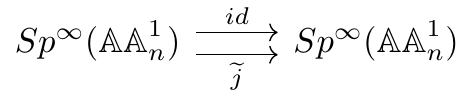}, \end{center}
where $\widetilde{j}$ runs over all maps induced on $Sp^{\infty}(\A_n^1)$ by maps of the form \[Sp^m(\A_n^1) \longrightarrow Sp^{m}(\A_n^1) \times * \xlongrightarrow{id \otimes j^*} Sp^{m}(\A_n^1) \times Sp^{1}(\A_n^1) \xlongrightarrow{+} Sp^{m+1}(\A_n^1),\] where $j^*$ is the dual of a map $j : \R[\xx, \yy] \to \R$ such that $j(x_i) = j(y_i)$ for some $i$.

Since $\widetilde{Sp}(\A_n^1)$ is a quotient of $Sp^\infty(\A_n^1)$, we have that \[A\left[\widetilde{Sp}(\A_n^1)\right] \subseteq A\left[Sp^\infty(\A_n^1)\right].\] It is our goal to investigate $A\left[\widetilde{Sp}(\A_n^1)\right]$ and gain insight into its structure. Let us say, perhaps somewhat preemptively, that the reader may identify $\widetilde{Sp}(\R^n \times \R^n)$ with the ``finite" $\R$-points of $\widetilde{Sp}(\A_n^1)$ -- those which can be induced by homomorphisms \[\R[\xx_i, \yy_i]_{1 \leq i \leq m} \to \R.\] The next section will justify this identification. For $X \in Sp^{\infty}(\R^n \times \R^n)$, we denote the associated $\R$-point of $Sp^{\infty}(\A_n^1)$ (or $\widetilde{Sp}(\A_n^1)$) by $\varphi_X^*$.

\subsection{Multisymmetric and Invariant Polynomials} \label{multisymmetric_section}

We review some facts about $A[Sp^\infty (\A_n^1)]$ from \cite{multi_symmetric}.

\begin{lemma}[\cite{multi_symmetric}]\label{multi_symmetric_lemma} 
$Sp^m(\A_n^1)$ and $Sp^{\infty}(\A_n^1)$ are affine schemes. Furthermore, $A[Sp^m(\A_n^1)] = A[\A_n^m]^{S_m}$ is generated (with relations) as an $\R$-algbera by the multi-symmetric power sums \[p_{\aa, \bb, m} = \sum_{i=1}^m x_{i1}^{a_1} x_{i2}^{a_{2}} \cdots x_{in}^{a_{n}}y_{i1}^{b_1} y_{i2}^{b_{2}} \cdots y_{in}^{b_{n}},\] where $a_j, b_j \in \N$. Additionally, $A\left[Sp^\infty(\A_n^1)\right]$ is the subset of the ring of power series $\R\llbracket p_{\aa, \bb}\rrbracket$, where \[p_{\aa, \bb} = \sum_{i=1}^\infty x_{i1}^{a_1} x_{i2}^{a_{2}} \cdots x_{in}^{a_{n}}y_{i1}^{b_1} y_{i2}^{b_{2}} \cdots y_{in}^{b_{n}},\] consisting of power series $p$ with the property that, given any finite subset $S$ of $\{p_{\aa, \bb}\}$, there are finitely many terms in $p$ involving only elements of $S$. In particular, there are no relations among the $p_{\aa, \bb}$ in $A\left[Sp^\infty(\A_n^1)\right]$.
\end{lemma}

Unfortunately, the ring $A\left[Sp^\infty(\A_n^1)\right]$ is too infinite for our purposes. We wish to isolate a tractable yet still suitably large subring $A_{fin}\left[Sp^\infty(\A_n^1)\right]$ of $A\left[Sp^\infty(\A_n^1)\right]$ so that if we are given $X \in Sp^{\infty}(\R^n \times \R^n)$ and $f \in A_{fin}\left[Sp^\infty(\A_n^1)\right]$, the quantity $\varphi_X(f)$ is guaranteed to be finite.

To achieve this goal, it will be necessary to work in the category of graded rings. We give all variables $x_{ij}$ and $y_{ij}$ the grading $1$. Define \[A_{fin}\left[Sp^\infty(\A_n^1)\right] = \varprojlim_{m} A\left[Sp^m(\A_n^1)\right],\] where this inverse limit is taken in the category of graded rings. If we forget the gradings, we see that $A_{fin}\left[Sp^\infty(\A_n^1)\right] \subseteq A\left[Sp^\infty(\A_n^1)\right]$ is simply the polynomial ring $\R\left[p_{\aa, \bb}\right]$ (because $A\left[Sp^\infty(\A_n^1)\right]$ has only finitely many generators in each degree). Due to the interpretation of $A_{fin}\left[Sp^\infty(\A_n^1)\right]$ as an inverse limit in the category of graded rings, we call $A_{fin}\left[Sp^\infty(\A_n^1)\right]$ the $K$-finite elements of $A\left[Sp^\infty(\A_n^1)\right]$.

Our goal is to isolate the elements $f$ of $A_{fin}\left[Sp^\infty(\A_n^1)\right]$ which ``disregard" any cubes with volume $0$ in the following sense: for all $X_1, X_2 \in Sp^\infty(\R^n \times \R^n)$ such that $X_1$ and $X_2$ are identified in $\widetilde{Sp}(\R^n \times \R^n)$, we require that $\varphi_{X_1}(f) = \varphi_{X_2}(f)$ (where $\varphi_{X_1}^*$ and $\varphi_{X_2}^*$ are the $\R$-points of $\widetilde{Sp}(\A_n^1)$ associated with $X_1$ and $X_2$). We denote the collection of such $f$ by $A_{fin}\left[\widetilde{Sp}(\A_n^1)\right]$.

Let $W_{i} \subseteq (\R^n \times \R^n)^m$ be the closed subset of all points which satisfy the equation \[\prod_{k=1}^n y_{ik} - x_{ik} = 0.\] Note that $W_i$ is the union of $n$ hyperplanes in $\R^n \times \R^n$. To any point $X \in W_i$, we can associate an $\R$-point $\varphi_X^*$ of $\A_n^m$. Consider the subring $R_n^m$ of $A[\A_n^m]$ consisting of all $f \in A[\A_n^m]$ such that for all $i$ and for all $X_1, X_2 \in W_i$, $\varphi_{X_1}(f) = \varphi_{X_2}(f)$. Put differently, $R_n^m$ is the subring of $A[\A_n^m]$ consisting of all polynomials \[f \in \R[x_{ij},\ y_{ij}]_{1 \leq i \leq m,\ 1 \leq j \leq n}\] such that for all $i$ and $j$, the restriction of $f$ to $W_i$ is independent of $x_{ij}$ and $y_{ij}$. Let $R_n^\infty = \varprojlim_m R_n^m$ (where the inverse limit is taken in the category of graded rings).

\begin{proposition}\label{sp_inv_lim}
Taking inverse limits in the category of graded rings,
\begin{align*}
A_{fin}\left[\widetilde{Sp}(\A_n^1)\right] &= R_n^\infty \cap A_{fin}\left[Sp^\infty(\A_n^1)\right] \\
&= \left(\varprojlim_m R_n^m\right) \cap \left(\varprojlim_m A\left[Sp^m(\A_n^1)\right] \right) \\
&= \varprojlim_m \left( R_n^m \cap A\left[ Sp^m(\A_n^1) \right] \right).
\end{align*}
\end{proposition}

\subsection{The Calculation}

In this section, we examine the structure and properties of $R_n^m$ and $A_{fin}\left[\widetilde{Sp}(\A_n^1)\right]$. It will be convenient now to change coordinates. Define \[\eta_i = \yy_i - \xx_i \ \ \textup{and} \ \ \xi_i = \yy_i + \xx_i. \] Note that \[A[\A_n^m] = \R[\eta_{ij},\ \xi_{ij}]_{1 \leq i \leq m,\ 1 \leq j \leq n}.\] 

\begin{proposition} \label{diff_char}
The ring $R_n^m$ is characterized algebraically as the subring of all $f \in A[\A_n^m]$ such that $\frac{\partial f}{\partial \xi_{ij}}\in (\eta_{ik})$ for all $i, j, k$, and such that $\frac{\partial f}{\partial \eta_{ij}}\in (\eta_{ik})$ for all $i, j, k$ with $j \neq k$.
\end{proposition}
\begin{proof}
With our new coordinates, $W_i$ can be defined by the equation $\prod_{k=1}^n \eta_{ik} = 0$. Let $W_{ik} \subseteq \R^n \times \R^n$ denote the hyperplane defined by the equation $\eta_{ik} = 0$, and note that $W_i = \bigcup_{k=1}^n W_{ik}$. Then $R_n^m$ is the subring of $A[\A_n^m]$ consisting of all functions whose restriction to $W_i$ is independent of $\xi_{ij} - \eta_{ij}$ and $\xi_{ij} + \eta_{ij}$ for all $i,j$. Equivalently, $R_n^m$ is the subring of $A[\A_n^m]$ consisting of all functions whose restriction to $W_{ik}$ is independent of $\xi_{ij}$ and $\eta_{ij}$ for all $i,j,k$. 

Suppose $f \in \R[\xi_{ij},\, \eta_{ij}]$. For any $i,k$, we can write \[f(\xi, \eta) = q(\xi, \eta) \eta_{ik} + r_{ik}(\xi, \eta),\] where $\eta_{ik}$ does not appear in $r_{ik}(\xi, \eta)$. Since $f|_{W_{ik}} = r|_{W_{ik}}$, it follows that $f \in R_n^m$ if and only if for all $i,j,k$, $\frac{\partial r_{ik}}{\partial \xi_{ij}} = \frac{\partial r_{ik}}{\partial \eta_{ij}} = 0$. This in turn is equivalent to the condition that $\frac{\partial f}{\partial \xi_{ij}}\in (\eta_{ik})$ for all $i, j, k$ and $\frac{\partial f}{\partial \eta_{ij}}\in (\eta_{ik})$ for all $i, j, k$ with $j \neq k$.
\end{proof}

\begin{proposition}
The following set of monomials forms an $\R$-basis for $R_n^m$: \[\left\{\prod_{i,j} \eta_{ij}^{a_{ij}} \xi_{ij}^{b_{ij}}\ \middle|\ \left(\forall i\right)\left(\left(\exists j \textup{ such that } a_{ij} > 0 \textup{ or } b_{ij} > 0\right) \Longrightarrow \left(\forall k, a_{ik} > 0\right)\right)\right\}.\]
\end{proposition}
\begin{proof}
The differential conditions given in proposition \ref{diff_char} which are necessary and sufficient for a polynomial $f \in A[\A_n^m]$ to belong to $R_n^m$ are true for a polynomial $f$ if and only if they are true for each of the monomial summands of $f$. Hence, we know that there exists a basis of monomials for $R_n^m$. The monomials which satisfy the conditions of proposition \ref{diff_char} are exactly those listed in the statement of this proposition.
\end{proof}

Having completed these initial calculations, it remains to calculate $A_{fin}\left[\widetilde{Sp}(\A_n^1)\right]$. Recall that the Hilbert series $HS_\Omega (t)$ of a real vector space $\Omega$ is defined to be the sum \[HS_\Omega(t) = \sum_{i=1}^\infty \dim_{\R}(\Omega)\, t^i.\]

\begin{theorem}\label{algebraic_theorem}
$A_{fin}\left[\widetilde{Sp}(\A_n^1)\right]$ is freely generated as an $\R$-algebra by the infinite symmetric polynomials \[p_{\aa, \bb} = \sum_{i=1}^\infty \eta_{i1}^{a_1} \eta_{i2}^{a_{2}} \cdots \eta_{in}^{a_{n}}\xi_{i1}^{b_1} \xi_{i2}^{b_{2}} \cdots \xi_{in}^{b_{n}}\] where $a_j \geq 1$ for all $j$.
\end{theorem}
\begin{proof}
Let $\calA$ denote the $\R$-algebra generated by the $p_{\aa, \bb}$ such that $a_j \geq 1$ for all $j$. By lemma \ref{multi_symmetric_lemma}, $\calA$ is freely generated by these $p_{\aa, \bb}$. Since each of these $p_{\aa, \bb}$ is in $A_{fin}\left[\widetilde{Sp}(\A_n^1)\right]$, we have that $\calA$ is a freely generated sub-algebra of $A_{fin}\left[\widetilde{Sp}(\A_n^1)\right]$. It remains to show that these two algebras are in fact equal.

To prove equality, we need only show that $HS_{\calA}(t) = HS_{A_{fin}\left[\widetilde{Sp}(\A_n^1)\right]}(t)$. This equality follows from the results of lemmas \ref{subalgebra-lemma} and \ref{algebra-lemma} below.
\end{proof}

\begin{lemma}\label{subalgebra-lemma}
\[HS_{\calA}(t) = \prod_{d = 0}^\infty \left(1 - t^{n+d}\right)^{-\binom{d+2n-1}{2n-1}}.\]
\end{lemma}
\begin{proof}
Since we require that $a_j \geq 1$ for all generators $p_{\aa, \bb}$ of $\calA$, it follows that $\calA$ has $0$ generators in degrees $1$ through $n-1$. In degrees $d \geq n$, $\calA$ has $\binom{d+n-1}{d-n} = \binom{d+n-1}{2n-1}$ generators (this multinomial coefficient represents the number of ways to put $d-n$ balls into $2n$ buckets).
\end{proof}

\begin{lemma}\label{algebra-lemma}
\[HS_{A_{fin}\left[\widetilde{Sp}(\A_n^1)\right]}(t) = \prod_{d = 0}^\infty \left(1 - t^{n+d}\right)^{-\binom{d+2n-1}{2n-1}}.\]
\end{lemma}
\begin{proof}
For ease of notation, let $f(t) = HS_{A_{fin}\left[\widetilde{Sp}(\A_n^1)\right]}(t)$. Additionally, we define $g(t) = \prod_{d = 0}^\infty \left(1 - t^{n+d}\right)^{-\binom{d+2n-1}{2n-1}}$. We must show that $f(t) = g(t)$. 

Let $<_{\lex}$ denote the lexicographic order on $\N_+^n \times \N^n$. Let $\prec$ denote the following linear order on $\N_+^n \times \N^n$: for $(\aa, \bb), (\aa^\prime, \bb^\prime) \in \N_+^n \times \N^n$, $(\aa^\prime, \bb^\prime) \prec (\aa, \bb)$ if either \[\aa^\prime + \bb^\prime < \aa + \bb,\] or \[\aa^\prime + \bb^\prime = \aa + \bb \textup{  and  } (\aa^\prime, \bb^\prime) <_{\lex} (\aa, \bb).\] 


Let $(\aa, \bb) \in \N_+^n \times \N^n$. Define $f_{\aa, \bb}(t)$ to be the Hilbert series of the free subalgebra of $A_{fin}\left[\widetilde{Sp}(\A_n^1)\right]$ generated by the symmetrizations of the monomials \[\prod_{i = 1}^l \prod_{j=1}^{n} \eta_{ij}^{a_{ij}} \xi_{ij}^{b_{ij}} \textup{  such that  } (\aa, \bb) \succeq (\aa_1, \bb_1) \succeq (\aa_2, \bb_2) \succeq \cdots \succeq (\aa_l, \bb_l),\] where $(\aa_i, \bb_i) \in \N_+^n \times \N^n$. For $(\aa, \bb) \in \N_+^n \times \N^n$, let $(\aa^\prime, \bb^\prime) \in \N_+^n \times \N^n$ be the immediate predecessor to $(\aa, \bb)$ under the $\preceq$ ordering (if it exists). Then we have \[f_{\mathbf{1}, \mathbf{0}}(t) = (1 - t^n)^{-1}\] and \[ f_{\aa, \bb} (t) = (1 - t^{\aa + \bb})^{-1} f_{\aa^\prime, \bb^\prime} (t).\]

We now analyze the function $g(t)$. For \[D_{\aa, \bb}(t) = \prod_{0 \leq d < \aa + \bb - n} \left(1-t^{d+n}\right)^{-\binom{d+2n-1}{2n-1}},\] \[ A_{\aa, \bb, i} (t) = \prod_{0 \leq \alpha < a_i - 1} \left(1-t^{\aa+\bb}\right)^{-\binom{n+\aa+\bb - \left(\sum_{j < i} a_j\right) - \alpha - 2}{2n-i-1}},\] and \[B_{\aa, \bb, i}(t) = \prod_{0 \leq \beta < b_i} \left(1-t^{\aa+\bb}\right)^{-\binom{n+\bb-i-\left(\sum_{j < i} b_j\right) - \beta - 1}{n-i-1}},\] set \[g_{\aa, \bb}(t) = \left(1-t^{\aa+\bb}\right)^{-1}D_{\aa, \bb}(t) \left(\prod_{i=1}^{n}A_{\aa, \bb, i}(t)\right)\left(\prod_{i=1}^{n-1}B_{\aa, \bb, i}(t)\right).\] Note that \[g_{\mathbf{1}, \mathbf{0}}(t) = (1 - t^n)^{-1}.\] Lemma \ref{super-icky} below shows that \[ g_{\aa, \bb} (t) = (1 - t^{\aa + \bb})^{-1} g_{\aa^\prime, \bb^\prime} (t).\]

Since $f_{\mathbf{1}, \mathbf{0}} = g_{\mathbf{1}, \mathbf{0}}$ and the $f_{\aa, \bb}$ and $g_{\aa, \bb}$ satisfy the same recurrence relation, it follows that $f_{\aa, \bb} = g_{\aa, \bb}$ for each $\aa, \bb$, and thus \[f(t) = \lim_{(\aa, \bb) \to \infty} f_{\aa, \bb}(t) = \lim_{(\aa, \bb) \to \infty} g_{\aa, \bb}(t) = g(t).\]
\end{proof}

\begin{lemma}\label{super-icky}
In the notation given in lemma \ref{algebra-lemma}, the functions $g_{\aa, \bb}$ satisfy \[ g_{\aa, \bb} (t) = (1 - t^{\aa + \bb})^{-1} g_{\aa^\prime, \bb^\prime} (t).\]
\end{lemma}
\begin{proof}
Due to the complicated definition of the total order $\prec$ used above, this lemma must be proven separately for the following five cases:
\begin{enumerate}
  \item $\aa^\prime = \aa$, ${b^\prime}_i = b_i$ for $1 \leq i \leq n-2$, ${b^\prime}_{n-1} + 1 = b_{n-1}$, and ${b^\prime}_{n-1} -1 = b_{n-1}$.
  \item $\aa^\prime = \aa$, and there exists an index $i_0 < n-1$ such that ${b^\prime}_i = b_i$ for $i < i_0$, ${b^\prime}_{i_0} + 1 = b_{i_0}$, ${b^\prime}_{i_0 + 1} - 1 = b_n$, ${b^\prime}_n = b_{i_0 + 1} = 0$, and ${b^\prime}_i = b_i = 0$ for $i_0 + 2 \leq i \leq n-1$.
  \item ${a^\prime}_i = a_i$ for $i < n$, ${a^\prime}_n + 1= a_n$, ${b^\prime}_{1} - 1 = b_{n}$, ${b^\prime}_n = b_{1} = 0$, and ${b^\prime}_i = b_i = 0$ for $2 \leq i \leq n-1$.
  \item there exists an index $i_0 \leq n-1$ such that ${a^\prime}_i = a_i$ for $i < i_0$, ${a^\prime}_{i_0} + 1 = a_{i_0}$, ${a^\prime}_{i_0 + 1} - 2 = b_n$, $a_{i_0 + 1} = 1$, ${a^\prime}_i = a_i = 1$ for $i \geq i_0 + 2$, $\bb^\prime = 0$, and $b_i = 0$ for all $i < n$.
  \item ${a^\prime}_1 = b_n$, $a_1 = 1$, ${a^\prime}_i = a_i = 1$ for all $i > 1$, ${b^\prime}_i = b_i = 0$ for $i < n$, and ${b^\prime}_n = 0$.
\end{enumerate}

For case (1), we have: \[(\aa^\prime, \bb^\prime) = \left(\left(\begin{array}{c} a_1 \\ \vdots \\ a_{n-1} \\ a_n \end{array}\right), \left(\begin{array}{c} b_1 \\ \vdots \\ b_{n-1}-1 \\ b_n + 1 \end{array}\right)\right) \textup{ and }(\aa, \bb) = \left(\left(\begin{array}{c} a_1 \\ \vdots \\ a_{n-1} \\ a_n \end{array}\right), \left(\begin{array}{c} b_1 \\ \vdots \\ b_{n-1} \\ b_n \end{array}\right)\right).\] Note that $B_{\aa, \bb, n-1}(t) = \left(1-t^{\aa+\bb}\right)^{-1}B_{\aa^\prime, \bb^\prime, n-1}(t)$. As all other factors of $g_{\aa, \bb}$ and $g_{\aa^\prime, \bb^\prime}$ are identical, the lemma holds.

The proofs of cases (2) through (5) are very similar. Of these, for ease, we only prove case (2). For this case, we have: \[(\aa^\prime, \bb^\prime) = \left(\left(\begin{array}{c} a_1 \\ \vdots \\ a_{i-1} \\ a_i \\ a_{i+1} \\ \vdots \\ a_n \end{array}\right), \left(\begin{array}{c} b_1 \\ \vdots \\ b_{i-1}-1 \\ b_n+1 \\ 0 \\ \vdots \\ 0 \end{array}\right)\right) \textup{ and }(\aa, \bb) = \left(\left(\begin{array}{c} a_1 \\ \vdots \\ a_{i-1} \\ a_i \\ a_{i+1} \\ \vdots \\ a_n \end{array}\right), \left(\begin{array}{c} b_1 \\ \vdots \\ b_{i-1} \\ 0 \\ 0 \\ \vdots \\ b_n \end{array}\right)\right).\] It suffices to show that \[B_{\aa, \bb, i-1}(t) = B_{\aa^\prime, \bb^\prime, i-1}(t) B_{\aa^\prime, \bb^\prime, i}(t)\left(1-t^{\aa+\bb}\right)^{-1},\] as all other factors of $g_{\aa, \bb}$ and $g_{\aa^\prime, \bb^\prime}$ are identical. Both sides of the equation above are merely products of powers of $1 - t^{\aa+\bb}$, so we need only show that the exponents of $1 - t^{\aa+\bb}$ are the same on both sides of the equality. That is, we must show that \[\binom{n+\bb-(i-1)-\left(\sum_{j < i-1}b_j\right) - (b_{i-1}-1) - 1}{n-(i-1)-1} = 1 + \sum_{\beta = 0}^{b_n} \binom{n+\bb-i-(\sum_{j < i} b_j)+1-\beta-1}{n-i-1}.\] This simplifies to \[\binom{n+\bb-i-\sum_{j < i}b_j+1}{n-i} = 1 + \sum_{\beta = 0}^{b_n} \binom{n+\bb-i-(\sum_{j < i} b_j)-\beta}{n-i-1}.\] Because $b_j = 0$ for $i+1 \leq j < n$, we can further simplify to \begin{align*} \binom{(n-i) + (b_n + 1)}{(n-i)} &= 1+\sum_{\beta=0}^{b_n} \binom{(n-i) + (b_n - \beta)}{(n-i)-1} \\ &= \sum_{\beta=0}^{b_n+1} \binom{(n-i) + (b_n - \beta)}{(n-i)-1} \\ &= \sum_{\beta^\prime=0}^{b_n+1} \binom{(n-i) + \beta^\prime - 1}{(n-i)-1}. \end{align*} This equality holds due to the combinatorial identity \[\binom{x+k}{x} = \sum_{k^\prime = 0}^k \binom{x+k^\prime-1}{x-1},\] which is easily proven by induction on $k$ using Pascal's rule.
\end{proof}

\section{Extending and Calculating the Invariants}\label{calculate-invariants}

Having determined the ring of $K$-finite global sections of the ``cubical" persistence modules $\calR(k,n)$, we are now faced with two problems. First, it is not clear how (or if) these functions might extend to the class of arbitrary persistence modules $\calM(k,n)$. Second, although we have determined the ring of $K$-finite global sections of $\calR(k,n)$, it is not clear how one might calculate them. This section addresses both of these issues.

\subsection{Defining the Invariants} \label{defining_invariants}

For a multidimensional persistence module $M$ over a field $k$ and $\uu, \vv \in \R^n$, the rank invariant $\rho_{\uu, \vv}(M)$ is defined by \[\rho_{\uu, \vv}(M) = \begin{cases} \dim_{k} \left(\im\left(M_{\lfloor \uu \rfloor} \to M_{\lceil \vv \rceil}\right)\right) & \uu \leq \vv \\ 0 & \textup{otherwise} \end{cases}.\] For a real (rather than integral) treatment of multidimensional persistence, one would remove the floor and ceiling symbols in the definition of $\rho_{\uu, \vv}$. An algorithm to calculate the rank invariant is given in \cite{computing_multid_persistence}.



For $\aa \in \N_+^n$, let $I = \{i\ |\ a_i = 1\}$, and let $J = \{i\ |\ a_i > 1\}$. Let $(\zz, \zz^\prime) \in \R^{n + |J|}$ denote a variable constructed so that $\zz, \zz^\prime \in \R^n$ denote variables such that for all $i \in I$, the $i$th coordinate of $\zz$ is the same as the $i$th coordinate of $\zz^\prime$.

As an example of the above construction, if we have $n=3$ and $\aa = (1, 2, 1)$, then $I = \{1, 3\}$, $J = \{2\}$, and $(\zz, \zz^\prime) \in \R^4$ where $\zz = (z_1, z_2, z_3)$ and $\zz^\prime = (z_1, z_2^\prime, z_3)$.

For any persistence module $M$ and $(\aa, \bb) \in \N_+^n \times \N^n$, define \[F_{\aa, \bb}(M) = \int_{\R^{n+|J|}} \left(\zz^\prime - \zz\right)^{\aa_J - 2}\left(\zz\right)^{\bb_I} \left(\zz + \zz^\prime\right)^{\bb_J}  \left(\rho_{\zz, \zz^\prime}(M)\right) \, d(\zz, \zz^\prime).\]

\subsection{Proving Equivalence of Invariants}

In this section, we will prove that on $\calR(k,n)$, the invariants $\{F_{\aa, \bb}\}$ defined in Section \ref{defining_invariants} are equivalent to the invariants $\{p_{\aa, \bb}\}$ defined in theorem \ref{algebraic_theorem} (equivalent in the sense that they encode the same information about an element $M \in \calR(k,n)$). This shows that the invariants $\{F_{\aa, \bb}\}$ extend to all of $\calM(k,n)$ the invariants $\{p_{\aa, \bb}\}$ defined on $\calR(k,n)$.

More concretely, we will prove that on $\calR(k,n)$, \[\Span\left(\left\{ F_{\aa, \bb}\right\}_{(\aa, \bb) \in \N_+^n \times \N^n}\right) = \Span\left(\left\{ p_{\aa, \bb}\right\}_{(\aa, \bb) \in \N_+^n \times \N^n}\right).\] We prove this fact in three iterations - first for $\calR^1(k,1)$, then for $\calR^1(k,n)$, and finally for $\calR(k,n)$.

In proving equivalence, it will be necessary to define a partial order $\preceq$ on $\N_+^n \times \N^n$ (which is different than the linear order $\preceq$ defined in the proof of lemma \ref{algebra-lemma}). For $(\aa, \bb), (\aa^\prime, \bb^\prime) \in \N_+^n \times \N^n$, we say that $(\aa, \bb) \preceq (\aa^\prime, \bb^\prime)$ if and only if \[\aa + \bb = \aa^\prime + \bb^\prime \textup{  and  } \aa \leq \aa^\prime.\] We extend this partial order to monomials: $\xx^{\aa}\yy^{\bb} \preceq \xx^{\aa^{\prime}}\yy^{\bb^{\prime}}$ if and only if $(\aa, \bb) \preceq (\aa^\prime, \bb^\prime)$, and we consider monomials $\xx^{\aa^{\prime}}\yy^{\bb^{\prime}}$ with $\xx^{\aa^{\prime}}\yy^{\bb^{\prime}} \succeq \xx^{\aa}\yy^{\bb}$ as ``higher order terms" (it will be clear from context what the monomial represented here by $\xx^{\aa}\yy^{\bb}$ is).

\begin{lemma}\label{span-1-1}
On $\calR^1(k,1)$, \[\Span\left(\left\{ F_{a, b}\right\}_{(a, b) \in \N_+ \times \N}\right) = \Span\left(\left\{ p_{a, b}\right\}_{(a, b) \in \N_+ \times \N}\right).\]
\end{lemma}
\begin{proof}
It suffices to prove that for all $k$, the following holds (on $\calR^1(k,1)$): \[\Span\left(\left\{ F_{a, b}\right\}_{a+b=k}\right) = \Span\left(\left\{ p_{a,b}\right\}_{a+b=k}\right).\] More concretely, we will show that $F_{a,b}$ can be written as a (finite) linear combination of $p_{a,b}$ and higher order terms. It follows from this computation that \[\Span\left(\left\{ F_{a, b}\right\}_{a+b=k}\right) \subseteq \Span\left(\left\{ p_{a,b}\right\}_{a+b=k}\right).\] Since there are only finitely many pairs $(a^\prime, b^\prime) \in \N_+ \times \N$ with $a^\prime + b^\prime = a+b = k$, it follows from back substitution that \[\Span\left(\left\{ F_{a, b}\right\}_{a+b=k}\right) \supseteq \Span\left(\left\{ p_{a,b}\right\}_{a+b=k}\right),\] effectively proving the lemma.

We return now to the matter of calculating $F_{a,b}(M)$ for $M \in \calR^1(k,1)$. Assume $M = [x, y] \in \calR^1(k,1)$. Let $(a, b) \in \N_+ \times \N$. We first consider the case $a = 1$. In this case, we have that $\rho_{z,z}(M)$ is simply the characteristic function $\one_{[x,y]}$, and so:
\begin{align*}
F_{1, b}(M) &= \int_{\R^{1}} z^b \left(\rho_{z,z}(M)\right) \, dz \\
&= \int_x^y z^b \, dz \\
&= \frac{1}{b+1} \left( y^{b+1} - x^{b+1} \right) \\
&= \left( \frac{1}{b+1} \right) \left(\left(\frac{(y+x)+(y-x)}{2}\right)^{b+1} - \left(\frac{(y+x)-(y-x)}{2}\right)^{b+1}\right) \\
&= \left(\frac{1}{(b+1)2^{b+1}}\right)\left( \sum_{i=0}^{b+1} \left(1-(-1)^i \right)\binom{b+1}{i}(y-x)^i(y+x)^{b+1-i} \right) \\
&= \left( \frac{1}{2^{b}} \right)p_{1,b}(M) + \HOT. 
\end{align*}

Next, we calculate $F_{a, b}(M)$ for $M = [x, y] \in \calR^1(k,1)$ and $(a, b) \in \N_+ \times \N$ with $a \geq 2$. Note that in this case, $\rho_{z,z^{\prime}}(M)$ is the characteristic function of the solid triangular region \[T_{x, y} = \left\{ (z, z^{\prime}) \in \R^2\ \middle|\ (z, z^\prime) \in [x, y] \times [x, y] \textup{ and } z^\prime \geq z \right\}.\] In the calculation of $F_{a,b}$, we will change variables by putting $\alpha = z^\prime - z$ and $\beta = z^\prime + z$. With these coordinates, \[T_{x, y} = \left\{ (\alpha, \beta) \in \R^2\ \middle|\ 2x + \alpha \leq \beta \leq 2y - \alpha \textup{ and } 0 \leq \alpha \leq y-x \right\}.\] We calculate:
\begin{align*}
F_{a, b}(M) &= \int_{\R^{2}} (z^{\prime} - z)^{a-2}(z^{\prime} + z)^b \left(\rho_{z,z^{\prime}}(M)\right) \, d(z, z^{\prime}) \\
&= \frac{1}{2} \int_0^{y-x} \int_{2x+\alpha}^{2y-\alpha} \alpha^{a-2} \beta^b \, d\beta \, d\alpha \\
&= \frac{1}{2(b+1)} \int_0^{y-x} \alpha^{a-2}\left((2y-\alpha)^{b+1} - (2x + \alpha)^{b+1}\right)  \, d\alpha \\
&= \frac{1}{2(b+1)} \int_0^{y-x} \alpha^{a-2}(2y-\alpha)^{b+1} - \alpha^{a-2} (2x + \alpha)^{b+1} \, d\alpha \\
&= \frac{1}{2(b+1)}\left[ \sum_{i=0}^{b+1} \left( \frac{(a-2)!\, (b+1)!}{(a-1+i)!\, (b+1-i)!} \right) \alpha^{a-1+i} (2y-\alpha)^{b+1-i} \right. \\ & \qquad \qquad \qquad \qquad\ \left. - (-1)^i \left( \frac{(a-2)!\, (b+1)!}{(a-1+i)!\, (b+1-i)!} \right) \alpha^{a-1+i} (2x+\alpha)^{b+1-i}\right]_{\alpha = 0}^{y-x} \\
&= \sum_{i=0}^{b+1} \left(\frac{1-(-1)^i}{2}\right) \left( \frac{(a-2)!\, b!}{(a-1+i)!\, (b+1-i)!} \right) (y-x)^{a-1+i} (y+x)^{b+1-i} \\
&= \left(\frac{1}{a(a-1)}\right)p_{a,b}(M) + \HOT.
\end{align*}
\end{proof}

\begin{lemma}\label{span-1-n}
On $\calR^1(k,n)$, \[\Span\left(\left\{ F_{\aa, \bb}\right\}_{(\aa, \bb) \in \N_+^n \times \N^n}\right) = \Span\left(\left\{ p_{\aa, \bb}\right\}_{(\aa, \bb) \in \N_+^n \times \N^n}\right).\]
\end{lemma}
\begin{proof}
The proof of this lemma will be similar to that of lemma \ref{span-1-1}. That is, we prove that for all $k$, the following holds (on $\calR^1(k,n)$): \[\Span\left(\left\{ F_{\aa, \bb}\right\}_{\aa+\bb=k}\right) = \Span\left(\left\{ p_{\aa,\bb}\right\}_{\aa+\bb=k}\right)\] by showing that $F_{\aa,\bb}$ can be written as a (finite) linear combination of $p_{\aa,\bb}$ and higher order terms.

Let $M \in \calR^1(k,n)$. Write $M$ as a product $M = \prod_{i=1}^n [x_i, y_i]$. Note that $M$ is determined by vertices $\xx, \yy \in \R^n$ with $\xx \leq \yy$. Let $(\aa, \bb) \in \N_+^n \times \N^n$. As discussed previously, we let $I = \{i\ |\ a_i = 1\}$ and $J = \{i\ |\ a_i > 1\}$. We can write the rank function $\rho_{\zz, \zz^\prime}(M)$ as a convenient product: \[ \rho_{\zz, \zz^\prime}(M) = \left(\prod_{i \in I} \one_{[x_i, y_i]}\right)\left(\prod_{i \in J} \one_{T_{x_i, y_i}}\right), \] where the region $T_{x_i, y_i}$ is as defined in the proof of lemma $\ref{span-1-1}$. This allows us to reduce the calculation of $F_{\aa, \bb}(M)$ to the calculations performed in lemma \ref{span-1-1}:
\begin{align*}
F_{\aa, \bb}(M) &= \int_{\R^{n+|J|}} \left(\zz^\prime - \zz\right)^{\aa_J - 2}\left(\zz\right)^{\bb_I} \left(\zz + \zz^\prime\right)^{\bb_J}  \left(\rho_{\zz, \zz^\prime}(M)\right) \, d(\zz, \zz^\prime)\\
&= \left(\prod_{i \in I} \int_{\R^{1}} z_i^{b_i} \left(\one_{[x_i,y_i]}(z_i)\right) \, dz_i \right) \\ & \quad \qquad \left( \prod_{i \in J} \int_{\R^{2}} (z_i^{\prime} - z_i)^{a_i-2}(z_i^{\prime} + z_i)^{b_i} \left(\one_{T_{x_i, y_i}} (z_i, z_i^\prime)\right) \, d(z_i, z_i^{\prime}) \right) \\
&= \left(\prod_{i \in I} \left( \frac{1}{2^{b_i}} \right) (y_i - x_i)(y_i + x_i)^{b_i} + \HOT\right) \\ & \qquad \quad\left( \prod_{i \in J} \left(\frac{1}{a_i(a_i-1)}\right)(y_i-x_i)^{a_i}(y_i+x_i)^{b_i} + \HOT \right) \\
&=  \left(\prod_{i \in I} \frac{1}{2^{b_i}} \right) \left(\prod_{i \in J} \frac{1}{a_i(a_i-1)}\right) p_{\aa, \bb}(M) + \HOT.
\end{align*}
\end{proof}

The previous two lemmas allow us to prove the main theorem of this section:

\begin{theorem}\label{span-n-n}
On $\calR(k,n)$, \[\Span\left(\left\{ F_{\aa, \bb}\right\}_{(\aa, \bb) \in \N_+^n \times \N^n}\right) = \Span\left(\left\{ p_{\aa, \bb}\right\}_{(\aa, \bb) \in \N_+^n \times \N^n}\right).\]
\end{theorem}
\begin{proof}
The proof of this lemma will be similar to that of Lemma \ref{span-1-n}. That is, we prove that for all $k$, the following holds (on $\calR(k,n)$): \[\Span\left(\left\{ F_{\aa, \bb}\right\}_{\aa+\bb=k}\right) = \Span\left(\left\{ p_{\aa,\bb}\right\}_{\aa+\bb=k}\right)\] by showing that $F_{\aa,\bb}$ can be written as a (finite) linear combination of $p_{\aa,\bb}$ and higher order terms.

Each element $M$ of $\calR(k,n)$ can be viewed as the finite direct sum \[M = \bigoplus_k M_k,\] where $M_k \in \calR^1(k,n)$. Hence, for $(\aa, \bb) \in \N_+^n \times \N^n$, we have that \[\rho_{\zz, \zz^\prime}(M) = \sum_k \rho_{\zz, \zz^\prime}(M_k).\] This allows us to calculate $F_{\aa, \bb}(M)$ from the calculation of $F_{\aa, \bb}(M_k)$ from lemma \ref{span-1-n}:
\begin{align*}
F_{\aa, \bb}(M) &= \int_{\R^{n+|J|}} \left(\zz^\prime - \zz\right)^{\aa_J - 2}\left(\zz\right)^{\bb_I} \left(\zz + \zz^\prime\right)^{\bb_J}  \left(\rho_{\zz, \zz^\prime}(M)\right) \, d(\zz, \zz^\prime)\\
&= \int_{\R^{n+|J|}} \left(\zz^\prime - \zz\right)^{\aa_J - 2}\left(\zz\right)^{\bb_I} \left(\zz + \zz^\prime\right)^{\bb_J}  \left( \sum_k \rho_{\zz, \zz^\prime}(M_k)\right) \, d(\zz, \zz^\prime)\\
&= \sum_k \int_{\R^{n+|J|}} \left(\zz^\prime - \zz\right)^{\aa_J - 2}\left(\zz\right)^{\bb_I} \left(\zz + \zz^\prime\right)^{\bb_J}  \left(\rho_{\zz, \zz^\prime}(M_k)\right) \, d(\zz, \zz^\prime)\\
&= \sum_k \left(\left(\prod_{i \in I} \frac{1}{2^{b_i}} \right) \left(\prod_{i \in J} \frac{1}{a_i(a_i-1)}\right) p_{\aa, \bb}(M_k) + \HOT \right)\\
&= \left(\prod_{i \in I} \frac{1}{2^{b_i}} \right) \left(\prod_{i \in J} \frac{1}{a_i(a_i-1)}\right) p_{\aa, \bb}(M) + \HOT. \\
\end{align*}
\end{proof}

\section{Group Completions and Geometric Insights}
\label{group_completions_geo_insights}

Recall that we have defined spaces $Sp^{\infty}(\calJ) = \varinjlim Sp^m(\calJ)$ and $\widetilde{Sp}(\R^n \times \R^n)$ by \[\widetilde{Sp}(\R^n \times \R^n) = \frac{\coprod_{m} Sp^m(\R^n \times \R^n)}{\simeq}.\] Both $Sp^{\infty}(\R^n \times \R^n)$ and $\widetilde{Sp}(\R^n \times \R^n)$ are commutative monoids generated by $Sp^1(\R^n \times \R^n)$ with monoid operation given by the natural maps \[+ : Sp^m(\R^n \times \R^n) \otimes Sp^{m^\prime}(\R^n \times \R^n) \to Sp^{m+m^\prime}(\R^n \times \R^n).\]

This section is dedicated to the definition and properties of the group completions $K\left(Sp^{\infty}(\R^n \times \R^n)\right)$ and $K\left(\widetilde{Sp}(\R^n \times \R^n)\right)$ of these monoids. We show that there is a bijection between elements of $K\left(Sp^{\infty}(\R^n \times \R^n)\right)$ and (suitably defined) generalized rank invariants. This bijection provides an alternative method to finding $p_{\aa, \bb}(M)$ which is computationally faster and simpler than the measure-theoretic method proposed in Section \ref{calculate-invariants}. We also show that $A_{fin}\left[\widetilde{Sp}(\A_n^1)\right]$ separates the elements of $K\left(\widetilde{Sp}(\R^n \times \R^n)\right)$.

\subsection{$K\left(Sp^{\infty}(\R^n \times \R^n)\right)$ and $K\left(\widetilde{Sp}(\R^n \times \R^n)\right)$.}

For a (cancellative) monoid $M$, we define the group completion $K(M)$ of $M$ by \[K(M) = \frac{M \times M}{\simeq},\] where for $(s_1, t_1), (s_2, t_2) \in M \times M$, we have $(s_1, t_1) \simeq (s_2, t_2)$ if and only if $s_1 + t_2 = s_1 + t_1$. It is implied that the element $(s,t) \in K(M)$ should be thought of as ``$s-t$".

We have discussed in Lemma \ref{sp_lemma} and Remark \ref{sp_remark} that $Sp^{\infty}(\R^n \times \R^n)$ can be identified with a certain subset of the $\R$-points of $Sp^{\infty}(\A_n^1)$. We now discuss the extension of this identification to $K(Sp^{\infty}(\R^n \times \R^n))$.

For $X \in Sp^{\infty}(\R^n \times \R^n)$, we may associate an $\R$-point $\varphi_X^*$ of $Sp^{\infty}(\A_n^1)$. The underlying homomorphism $\varphi_X : A_{fin}[Sp^{\infty}(\A_n^1)] \to \R$ has the potential of producing geometrically meaningful summaries of $X$, especially in the case where $X$ has arisen (in the sense of Corollary \ref{rect_to_sp}) from an element of $\calR(k,n)$. Indeed, if we fix some $f_i \in A[Sp^{\infty}(\A_n^1)]$, the numbers $\varphi_X(f_i)$ might provide insightful information about $X$. We now extend this methodology to the case $X \in K(Sp^{\infty}(\R^n \times \R^n))$ and $X \in K(\widetilde{Sp}(\R^n \times \R^n))$.

Concretely, for $(X_1, X_2) \in Sp^{\infty}(\R^n \times \R^n) \times Sp^{\infty}(\R^n \times \R^n)$, we define a ring homomorphism \[\varphi_{(X_1, X_2)} : A_{fin}[Sp^{\infty}(\A_n^1)] \to \R\] on the generators $p_{\aa, \bb}$ of $A_{fin}[Sp^{\infty}(\A_n^1)]$ (defined in Lemma \ref{multi_symmetric_lemma}) by \[\varphi_{(X_1, X_2)}(p_{\aa, \bb}) = \varphi_{X_1}(p_{\aa, \bb}) - \varphi_{X_2}(p_{\aa, \bb}).\]

\begin{proposition}\label{varphi_sp}
The map \[Sp^{\infty}(\R^n \times \R^n) \times Sp^{\infty}(\R^n \times \R^n) \longrightarrow \left( A_{fin}[Sp^{\infty}(\A_n^1)] \to \R \right)\] defined by \[(X_1, X_2) \mapsto \varphi_{(X_1, X_2)}\] factors through $K(Sp^{\infty}(\R^n \times \R^n))$, and can thus be viewed as a map \[K(Sp^{\infty}(\R^n \times \R^n)) \longrightarrow \left( A_{fin}[Sp^{\infty}(\A_n^1)] \to \R \right).\]
\end{proposition}
\begin{proof}
It suffices to observe that for $X_1, X_2, Y \in Sp^{\infty}(\R^n \times \R^n)$, the following equality holds (by definition): \[\varphi_{(X_1 + Y,\ X_2 + Y)} = \varphi_{(X_1, X_2)}.\] 
\end{proof}

By combining Theorem \ref{algebraic_theorem} and Proposition \ref{varphi_sp}, we obtain:

\begin{proposition}
The map \[Sp^{\infty}(\R^n \times \R^n) \longrightarrow \left( A_{fin}[Sp^{\infty}(\A_n^1)] \to \R \right)\] defined by \[X \mapsto \varphi_X\] descends to a map \[\widetilde{Sp}(\R^n \times \R^n) \longrightarrow \left( A_{fin}[\widetilde{Sp}(\A_n^1)] \to \R \right).\] Moreover, the map \[\widetilde{Sp}(\R^n \times \R^n) \times \widetilde{Sp}(\R^n \times \R^n) \longrightarrow \left( A_{fin}[\widetilde{Sp}(\A_n^1)] \to \R \right)\] defined by \[(X_1, X_2) \mapsto \varphi_{(X_1, X_2)}\] factors through $K(\widetilde{Sp}(\R^n \times \R^n))$, and can thus be viewed as a map \[K(\widetilde{Sp}(\R^n \times \R^n)) \longrightarrow \left( A_{fin}[\widetilde{Sp}(\A_n^1)] \to \R \right).\]
\end{proposition}

\subsection{Injectivity Lemmata}

Having defined maps $\varphi_X$ for $X$ in various spaces (e.g., $Sp^{\infty}(\R^n \times \R^n)$, $K(\widetilde{Sp}(\R^n \times \R^n))$), we now show that $\varphi_X \neq \varphi_Y$ given $X \neq Y$. These lemmata generalize a result of \cite{algebraic_functions}. 

\begin{lemma} \label{finite_injectivity}
Let $X,\ Y \in Sp^m(\R^n \times \R^n)$ such that $X \neq Y$. Then $\varphi_X \neq \varphi_Y$ (as maps $A[Sp^m(\A_n^1)] \to \R$).
\end{lemma}
\begin{proof}
Let $Z = \{(\zz_{i}, \ww_{i})\}_{i=1}^m$. Let $f \in A[Sp^m(\A_n^1)]$ be defined by \[f(\xx, \yy) = \prod_{\sigma \in S_m} \left( \sum_{i=1}^m \sum_{j = 1}^n \left( x_{\sigma(i),j} - z_{i, j} \right)^2 + \left( y_{\sigma(i),j} - w_{i, j} \right)^2 \right).\] Then for any $W \in Sp^m(\R^n \times \R^n)$, $\varphi_{W}(f) = 0$ if and only if $W = Z$.
\end{proof}

\begin{remark}
We have presented a concrete proof to Lemma~\ref{finite_injectivity}. However, this lemma also follows from the Nullstellensatz and the fact that for any real variety $V$, the affine coordinate ring $A[V]$ separates the (real) points of $V$.
\end{remark}

\begin{lemma} \label{infinite_injectivity}
Let $X,\ Y \in Sp^{\infty}(\R^n \times \R^n)$ such that $X \neq Y$. Then $\varphi_X \neq \varphi_Y$ (as maps $A_{fin}[Sp^{\infty}(\A_n^1)] \to \R$).
\end{lemma}
\begin{proof}
We may represent $X$ and $Y$ by elements $X_{fin},\ Y_{fin} \in Sp^m(\R^n \times \R^n)$. By Lemma~\ref{finite_injectivity}, there exists $f_{fin} \in A[Sp^m(\A_n^1)]$ such that $\varphi_{X_{fin}}(f_{fin}) \neq \varphi_{Y_{fin}}(f_{fin})$. By Lemma~\ref{multi_symmetric_lemma}, we may write $f_{fin}$ as a finite sum $f_{fin} = \sum_i p_{\aa_i, \bb_i, m}$. Define $f \in A_{fin}[Sp^{\infty}(\A_n^1)]$ by $f = \sum_i p_{\aa_i, \bb_i}$. Then \[\varphi_{X}(f) = \varphi_{X_{fin}}(f_{fin}) \neq \varphi_{Y_{fin}}(f_{fin}) = \varphi_{Y}(f).\]
\end{proof}

\begin{lemma} \label{group_infinite_injectivity}
Let $X,\ Y \in K(Sp^{\infty}(\R^n \times \R^n))$ such that $X \neq Y$. Then $\varphi_X \neq \varphi_Y$ (as maps $A_{fin}[Sp^{\infty}(\A_n^1)] \to \R$).
\end{lemma}
\begin{proof}
Choose coset representatives for $X$ and $Y$ in $Sp^{\infty}(\R^n \times \R^n) \times Sp^{\infty}(\R^n \times \R^n)$, and write $X = (X_+, X_-)$ and $Y = (Y_+, Y_-)$. Since $X_+ + Y_- \neq Y_+ + X_-$, there exists by Lemma~\ref{infinite_injectivity} some $f^\prime \in A_{fin}[Sp^{\infty}(\A_n^1)]$ such that $\varphi_{X_+ + Y_-}(f^{\prime}) \neq \varphi_{Y_+ + X_-}(f^{\prime})$.

Since $A_{fin}[Sp^{\infty}(\A_n^1)]$ is generated as an algebra by the power sums $p_{\aa, \bb}$, there exists a power sum $f = p_{\aa_0, \bb_0} \in A_{fin}[Sp^{\infty}(\A_n^1)]$ such that $\varphi_{X_+ + Y_-}(f) \neq \varphi_{Y_+ + X_-}(f)$. Because $f$ is a power sum, \[\varphi_{X_+ + Y_-}(f) = \varphi_{X_+}(f) + \varphi_{Y_-}(f) \qquad \textup{and} \qquad \varphi_{Y_+ + X_-}(f) = \varphi_{Y_+}(f) + \varphi_{X_-}(f).\] Combining these equalities and inequalities shows that \[\varphi_{X_+}(f) - \varphi_{X_-}(f) \neq \varphi_{Y_+}(f) - \varphi_{Y_-}(f),\] and hence
\begin{align*}
\varphi_{X}(f) &= \varphi_{(X_+, X_-)}(f) \\ &= \varphi_{X_+}(f) - \varphi_{X_-}(f) \\ &\neq \varphi_{Y_+}(f) - \varphi_{Y_-}(f) \\ &= \varphi_{(Y_+, Y_-)}(f) \\ &= \varphi_{Y}(f).
\end{align*}
\end{proof}

\begin{lemma} \label{sp_infinite_injectivity}
Let $X,\ Y \in \widetilde{Sp}(\R^n \times \R^n)$ such that $X \neq Y$. Then $\varphi_X \neq \varphi_Y$ (as maps $A_{fin}[\widetilde{Sp}(\A_n^1)] \to \R$).
\end{lemma}
\begin{proof}
By the definition of $\widetilde{Sp}(\R^n \times \R^n)$, we can represent any $Z \in \widetilde{Sp}(\R^n \times \R^n)$ by some $Z_{fin} \in Sp^{m^\prime}(\R^n \times \R^n)$ for some $m^\prime$. Choose $m$ minimal and $X_{fin},\ Y_{fin} \in Sp^{m}(\R^n \times \R^n)$ so that we can represent $X$ by $X_{fin}$ and $Y$ by $Y_{fin}$. By Proposition~\ref{sp_inv_lim}, it now suffices to produce some $f_{fin} \in R_n^m \cap A[Sp^m(\A_n^1)]$ such that $\varphi_{X_{fin}}(f_{fin}) \neq \varphi_{Y_{fin}}(f_{fin})$. Indeed, if $f \in A_{fin}[\widetilde{Sp}(\A_n^1)]$ restricts to $f_{fin}$, then \[\varphi_{X}(f) = \varphi_{X_{fin}}(f_{fin}) \neq \varphi_{Y_{fin}}(f_{fin}) = \varphi_{Y}(f).\]

Let \[g_{fin} = \prod_{i=1}^m \prod_{j=1}^n (y_{ij} - x_{ij}).\] Note that $g_{fin} \in A[Sp^m(\A_n^1)]$. By the minimality of $m$, it is not the case that $\varphi_{X_{fin}}(g_{fin}) = \varphi_{Y_{fin}}(g_{fin}) = 0$.

If $\varphi_{X_{fin}}(g_{fin}) \neq \varphi_{Y_{fin}}(g_{fin})$, then we may let $f_{fin} = g_{fin}$. On the other hand, if $\varphi_{X_{fin}}(g_{fin}) = \varphi_{Y_{fin}}(g_{fin}) \neq 0$, then by Lemma~\ref{finite_injectivity}, there exists $h_{fin} \in A[Sp^m(\A_n^1)]$ such that $\varphi_{X_{fin}}(h_{fin}) \neq \varphi_{Y_{fin}}(h_{fin})$. In this case, let $f_{fin} = g_{fin} h_{fin}$, and observe that
\begin{align*}
\varphi_{X_{fin}}(f_{fin}) &= \varphi_{X_{fin}}(g_{fin}h_{fin}) \\
&= \varphi_{X_{fin}}(g_{fin})\varphi_{X_{fin}}(h_{fin}) \\
&= \varphi_{Y_{fin}}(g_{fin})\varphi_{X_{fin}}(h_{fin}) \\
&\neq \varphi_{Y_{fin}}(g_{fin})\varphi_{Y_{fin}}(h_{fin}) \\
&= \varphi_{Y_{fin}}(g_{fin}h_{fin}) \\
&= \varphi_{Y_{fin}}(f_{fin}).
\end{align*}

Regardless of which definition we choose for $f_{fin}$, we have that $f_{fin} \in R_n^m$ by Proposition~\ref{diff_char} and that $f_{fin} \in A[Sp^m(\A_n^1)]$ because its constituent factors are elements of $A[Sp^m(\A_n^1)]$.
\end{proof}

\begin{lemma} \label{group_sp_infinite_injectivity}
Let $X,\ Y \in K(\widetilde{Sp}(\R^n \times \R^n))$ such that $X \neq Y$. Then $\varphi_X \neq \varphi_Y$ (as maps $A_{fin}[\widetilde{Sp}(\A_n^1)] \to \R$).
\end{lemma}
\begin{proof}
The proof of Lemma~\ref{group_sp_infinite_injectivity} from Lemma~\ref{sp_infinite_injectivity} is identical to the proof of Lemma~\ref{group_infinite_injectivity} from Lemma~\ref{infinite_injectivity}.
\end{proof}

\subsection{Generalized rank invariants.}
\label{generalized_rank_invariants}

The definition of the invariants $F_{\aa, \bb}$ of Section \ref{defining_invariants} relies quite heavily on the rank invariants $\rho_{\uu, \vv} : \calM(k,n) \to \N$. Having fixed $M \in \calM(k,n)$, we may view the rank invariant $\rho_{\uu, \vv}(M)$ as a function of its subscripts \[\rho_{-,-}(M) : \R^n \times \R^n \to \N\] of its subscripts. Note that $\rho_{\uu, \vv}(M) = 0$ unless $\uu \leq \vv$. Moreover, because we require elements of $\calM(k,n)$ to be finite, there exist $\uu_0, \vv_0 \in \Z^n$ such that for all $\uu, \vv \in \R^n$, we have that $\rho_{\uu, \vv}(M) = 0$ unless $\uu \geq \uu_0$ and $\vv \leq \vv_0$. Finally, even though we may evaluate $\rho_{\uu,\vv}(M)$ for arbitrary $\uu, \vv \in \R^n$, the values $\rho_{\uu,\vv}(M)$ are fully determined for $\uu, \vv \in \Z^n$.

These observations inspire the following definition:

\begin{definition}\label{generalized_rank}
A generalized rank invariant is a function $\rho_{-,-} : \Z^n \times \Z^n \to \Z$ such that:
\begin{enumerate}
\item $\rho_{\uu, \vv} = 0$ unless $\uu \leq \vv$.
\item There exist $\uu_0, \vv_0 \in \Z^n$ such that $\rho_{\uu, \vv} = 0$ for all $\uu, \vv \in \Z^n$ except if $\uu \geq \uu_0$ and $\vv \leq \vv_0$.
\end{enumerate}
\end{definition}

For the remainder of Section~\ref{generalized_rank_invariants}, we define \[J(n) = \left\{ (\xx, \yy) \in \Z^n \times \Z^n\ \middle| \xx \leq \yy \right\}.\] For $(\xx, \yy) \in J(n)$, we can define a generalized rank invariant $\rho_{-,-}((\xx, \yy))$ by \[ \rho_{\uu, \vv}((\xx, \yy)) = \begin{cases} 1 & \xx \leq \uu \leq \vv \leq \yy \\ 0 & \textup{otherwise} \end{cases}. \] We may symmetrize the above definition: for $X = \sum_{i=1}^m (\xx_i, \yy_i) \in Sp^{\infty}(J(n))$, define a generalized rank invariant $\rho_{-,-}(X)$ by \[\rho_{\uu, \vv}(X) = \sum_{i=1}^m \rho_{\uu, \vv}((\xx_i, \yy_i)).\] Finally, for $X \in K(Sp^{\infty}(J(n)))$ represented by $(X_+, X_-) \in Sp^{\infty}(J(n)) \times Sp^{\infty}(J(n))$, define a generalized rank invariant $\rho_{-,-}(X)$ by \[\rho_{\uu, \vv}(X) = \rho_{\uu, \vv}(X_+) - \rho_{\uu, \vv}(X_-).\]

\begin{lemma} \label{group_sp_rank_inj} 
The map from $K(Sp^{\infty}(J(n)))$ to the set of generalized rank invariants defined by $X \mapsto \rho_{-,-}(X)$ is injective.
\end{lemma}
\begin{proof}
Fix a total order $\preceq$ on $J(n)$ such that $(\xx, \yy) \preceq (\zz, \ww)$ if $\xx \leq \zz$ or if $\xx = \zz$ and $\yy \geq \ww$. 

Let $X \in K(Sp^{\infty}(J(n)))$ such that $\rho_{\uu,\vv}(X) = 0$ for all $(\uu, \vv) \in J(n)$. If $X \neq 0$, then we may write $X = \sum_{i = 1}^m c_i(\xx_i, \yy_i)$, where $(\xx_i, \yy_i) \prec (\xx_j, \yy_j)$ for $i < j$, $c_i \in \Z^*$, and $m \geq 1$. Then by the minimality of $(\xx_1, \yy_1)$ under the order $\preceq$, we have that $\rho_{\xx_1, \yy_1}(X) = c_1 \neq 0$, which is a contradiction.
\end{proof}

\begin{lemma} \label{group_sp_rank_surj} 
The map from $K(Sp^{\infty}(J(n)))$ to the set of generalized rank invariants defined by $X \mapsto \rho_{-,-}(X)$ is surjective.
\end{lemma}
\begin{proof}
As in the proof of Lemma~\ref{group_sp_rank_inj}, we fix a total order $\preceq$ on $J(n)$ such that $(\xx, \yy) \preceq (\zz, \ww)$ if $\xx \leq \zz$ or if $\xx = \zz$ and $\yy \geq \ww$. Let $\rho_{-,-}$ be a generalized rank invariant. We now describe an algorithm that produces an $X \in K(Sp^{\infty}(J(n)))$ such that $\rho_{-,-} = \rho_{-,-}(X)$.

If $\rho_{\uu, \vv} = 0$ for all $(\uu, \vv) \in \Z^n \times \Z^n$, let $X = 0$.

Otherwise, by conditions (1) and (2) of Definition~\ref{generalized_rank}, there exist $\uu_0, \vv_0 \in \Z^n$ such that $\rho_{\uu, \vv} = 0$ for all $\uu, \vv \in \Z^n$ except if $\uu_0 \leq \uu \leq \vv \leq \vv_0$. Note that the set all such $(\uu, \vv) \in \Z^n \times \Z^n$ such that $\uu_0 \leq \uu \leq \vv \leq \vv_0$ is finite. Choose $(\xx, \yy)$, minimal under the total order $\preceq$, such that $\rho_{\xx, \yy} \neq 0$. We complete the proof of Lemma~\ref{group_sp_rank_surj} by induction on \[\left|\left\{(\uu, \vv) \in \Z^n \times \Z^n\ \middle|\ (\xx, \yy) \preceq (\uu, \vv) \textup{ and } \uu_0 \leq \uu \leq \vv \leq \vv_0 \right\}\right|.\]

Define a generalized rank invariant $\rho^{\prime}_{-,-}$ by \[\rho^{\prime}_{-,-} = \rho_{-,-} - \rho_{-,-}((\xx, \yy)).\] Then $\rho_{\uu,\vv}^{\prime} = 0$ for all $(\uu, \vv) \preceq (\xx, \yy)$. Moreover, $\rho_{\uu, \vv} = 0$ for all $\uu, \vv \in \Z^n$ except if $\uu_0 \leq \uu \leq \vv \leq \vv_0$. By induction, there exists $X^{\prime} \in K(Sp^{\infty}(J(n)))$ such that $\rho^{\prime}_{-,-} = \rho_{-,-}(X^\prime)$. Let $X = X^{\prime} + \rho_{\xx, \yy} \cdot (\xx, \yy)$. Then $\rho_{-,-} = \rho_{-,-}(X)$.
\end{proof}

\begin{theorem} \label{group_sp_rank_bij} 
The map from $K(Sp^{\infty}(J(n)))$ to the set of generalized rank invariants defined by $X \mapsto \rho_{-,-}(X)$ is a bijection.
\end{theorem}
\begin{proof}
Injectivity follows Lemma~\ref{group_sp_rank_inj}. Surjectivity follows from Lemma~\ref{group_sp_rank_surj}.
\end{proof}

We may now use Theorem \ref{group_sp_rank_bij} to assist in the computation of $p_{\aa, \bb}(M)$ for arbitrary $M \in \calM(k,n)$. Recall that by Theorem \ref{span-n-n}, $p_{\aa, \bb}(M)$ is a linear combination of the invariants $F_{\aa^{\prime}, \bb^{\prime}}(M)$ (introduced in Section \ref{defining_invariants}), which are defined solely in terms of the rank invariants $\rho_{\uu, \vv}(M)$. That is, for $M \in \calM(k,n)$ and $X \in K(Sp^{\infty}(J(n)))$, we will have $p_{\aa, \bb}(M) = p_{\aa, \bb}(X)$ for all $\aa, \bb$ provided that $\rho_{\uu, \vv}(M) = \rho_{\uu, \vv}(X)$ for all $\uu, \vv$. 

Thus, to compute $p_{\aa, \bb}(M)$, we first use the rank invariants $\rho_{-,-}(M)$ and the algorithm presented in the proof of Lemma \ref{group_sp_rank_surj} to find some $X \in K(Sp^{\infty}(J(n)))$ such that $\rho_{-,-}(M) = \rho_{-,-}(X)$. We may now calculate $p_{\aa, \bb}(M)$ using the equality $p_{\aa, \bb}(M) = p_{\aa, \bb}(X)$. 

\begin{remark}
The results of this section also hold when working in the real (rather than integral) formulation of multidimensional persistence. However, in the real case, we would need to replace $\Z^n$ with $\R^n$ throughout this subsection and add a condition to the definition of a generalized rank invariant that ensures the finite termination of the algorithm in Lemma~\ref{group_sp_rank_surj}.
\end{remark}

\section{Recovering the Point Cloud}\label{recovering}

Section \ref{the-ring} examines the $K$-finite sections of $\calR(k,n)$, and Section \ref{calculate-invariants} provides a concrete method for calculating them. In this section, we provide a way to recover information about $M \in \calR(k,n)$ (including $M$ itself) from the invariants $p_{\aa, \bb}(M)$ (where $(\aa, \bb) \in \N_+^n \times \N^n$). In fact, these techniques will allow us to recover any $X \in \widetilde{Sp}(J(n))$ from the values $p_{\aa, \bb}(X)$.

\begin{remark}
We believe that an extension of the techniques used here will additionally allow us to recover any $X \in K(\widetilde{Sp}(\R^n \times \R^n))$. However, the algorithms and limits presented below are already sufficiently unwieldy to be computed in practice, and so we refrain from presenting the extended methodology.
\end{remark}

First, we prove two technical lemmas:

\begin{lemma}\label{silly-limit}
Let $a_1, ..., a_n \in \R$ with $a_j > 0$ for all $j$. Assume that $a_1 < 1$. Then \[\lim_{k \to \infty} \left(\prod_{j=1}^n a_j^{k^{n+1-j}} \right) = 0 .\]
\end{lemma}
\begin{proof}
We prove this lemma by induction on $n$. The case $n=1$ is trivial. If the lemma is true when $n = m-1$, it is also true for the case $n=m$, since \[\lim_{k \to \infty} \left(\prod_{j=1}^{m-1} a_j^{k^{(m-1)+1-j}}\right) = 0\] implies \[ \lim_{k \to \infty} \left(\prod_{j=1}^m a_j^{k^{m+1-j}} \right) = \lim_{k \to \infty} \left(a_m \left(\prod_{j=1}^{m-1} a_j^{k^{(m-1)+1-j}} \right)\right)^k = 0.\]
\end{proof}

\begin{lemma}\label{recovering-technical}
Let $w_i \in \R$ and $\zz_i \in \R^n$ for $1 \leq i \leq m$ such that $z_{i,j} > 0$ and $w_i > 0$. Further assume that the $\zz_i$ are in decreasing lexicographic order. Then \[\lim_{k \to \infty} \left( \sum_{i=1}^n \left( w_i \prod_{j=1}^n \left(\frac{z_{i,j}}{z_{1,j}}\right)^{k^{n+1-j}} \right) \right)^{1/k} = 1.\]
\end{lemma}
\begin{proof}
It suffices to show that there exist uniform (positive) lower and upper bounds on the quantity \[Q_k = \sum_{i=1}^n \left( w_i \prod_{j=1}^n \left(\frac{z_{i,j}}{z_{1,j}}\right)^{k^{n+1-j}} \right).\] Note that, as a sum of positive terms, $Q_k$ is greater than its first term $w_1$. It remains to show that the $Q_k$ are uniformly bounded from above.

Since $Q_k$ is the finite sum of terms of the form \[T_{k, i} = w_i \prod_{j=1}^n \left(\frac{z_{i,j}}{z_{1,j}}\right)^{k^{n+1-j}},\] it suffices to bound the $T_{k,i}$ uniformly from above. Because the $\zz_i$ are decreasingly lexicographically ordered, if $T_{k,i}$ is not identically equal to $w_i$, then there must be some least $j$, denoted $j_i$, such that $z_{i,j} < z_{1,j}$. Note that $\frac{z_{i, j_i}}{z_{1, j_i}} < 1$. Moreover, for all $j < j_i$, we must have $z_{i,j} = z_{1,j}$, and so \[T_{k, i} = w_i \prod_{j=i_j}^n \left(\frac{z_{i,j}}{z_{1,j}}\right)^{k^{n+1-j}} = w_i \prod_{j=1}^{n+1-i_j} \left(\frac{z_{i,j+i_j-1}}{z_{1,j+i_j-1}}\right)^{k^{n+2-i_j- j}}.\] This last term now falls under the auspices of lemma \ref{silly-limit}; thus, $\lim_{k \to \infty} T_{k,i} = 0$.
\end{proof}

The next theorem provides a means to recover information about the cubes which constitute some $M \in \calR(k,n)$ from the values $p_{\aa, \bb}(M)$.

\begin{theorem}\label{recovering-algebra}
Let $\calA$ be an algebra of nonnegative real functions on a set $\calX$. Suppose $f_1, \dots, f_n \in \calA$ and $x_1, \dots , x_m \in \calX$. Further assume that the $x_i$ are ordered so that the vectors $\{[f_j(x_i)]_{j \geq 1}\}_i \in \R^n$ are arranged in decreasing lexicographic order. Furthermore, assume that we only have access to the values \[\left\{\sum_{i=1}^m f(x_i)\right\}_{f \in \calA}.\] Then we can recover the set $\{f_j(x_i)\}$ inductively via the formula \[f_j(x_i) = \lim_{k \to \infty} \left( \frac{\sum_{i=1}^m \left( \prod_{j^\prime \leq j} \left(f_{j^\prime}(x_i)\right)^{k^{j+1-j^\prime}} \right) - \sum_{i^\prime < i} \left(\prod_{j^\prime \leq j} \left(f_{j^\prime}(x_i)\right)^{k^{j+1-j^\prime}}\right)}{\prod_{j^\prime < j} \left(f_{j^\prime}(x_i)\right)^{k^{j+1-j^\prime}}} \right)^{1/k}.\]
\end{theorem}

Letting $\calA$ equal the algebra generated by the $p_{\aa, \bb}$ for $(\aa, \bb) \in \N_+^n \times \N^n$ and letting $\calX$ equal $\calR^1(k,n)$ will allow us to recover the values of various functions on the individual cubes that constitute an element $M$ of $\calR(k,n)$. For example, letting $f_1 = p_{\mathbf{1}, \mathbf{0}}$ will allow us to determine in decreasing order the volumes of the cubes which constitute $M$.

\begin{proof}
Theorem \ref{recovering-algebra} follows from lemma \ref{recovering-technical} by setting $w_i = 1$ and $z_{i,j} = f_j (x_i)$.
\end{proof}

Theorem \ref{recovering-algebra} allows us to recover the values of any function in the ring of $K$-finite global sections of $\widetilde{Sp}(\R^n \times \R^n)$. However, the individual coordinate functions $\eta_i$ and $\xi_i$ are not, unfortunately, elements of this ring. The following variant of theorem \ref{recovering-algebra} provides a solution to this quandary.

\begin{theorem}
Let $M \in \calR(k,n)$, and assume that \[M = \bigoplus_{i=1}^m M_i,\] for $M_i \in \calR^1(k,n)$. Assume that the $M_i$ are ordered so that the vectors \[\left[ \begin{array}{c} p_{\mathbf{1}, \mathbf{0}}(M_i) \\ \eta_{j}(M_i) \\ \xi_j(M_i) \end{array} \right] \in \R^{2n+1}\] are arranged in decreasing lexicographic order. Let
\[f_{j}(M_i) = \left(p_{\mathbf{1}, \mathbf{0}}(M_i)\right)^{k^{j+1}} \left( \prod_{j^\prime \leq j} \left(\eta_{j^\prime}(M_i)\right)^{k^{j+1-j^\prime}} \right)\]
\[g_{j}(M_i) = \left(p_{\mathbf{1}, \mathbf{0}}(M_i)\right)^{k^{n+j+1}} \left( \prod_{j^\prime = 1}^n \left(\eta_{j^\prime}(M_i)\right)^{k^{n+j+1-j^\prime}}  \right)\left(\prod_{j^\prime \leq j} \left(\xi_{j^\prime}(M_i)\right)^{k^{j+1-j^\prime}} \right).\]

Then we can recover the values $\eta_j(M_i)$ and $\xi_j(M_i)$ inductively from the values that $A[\widetilde{Sp}(\R^n \times \R^n)]$ takes on $M$ via the following:

\[\eta_j(M_i) = \lim_{k \to \infty} \left(\frac{ \sum_{i^{\prime}=1}^m f_j(M_{i^\prime}) - \sum_{i^\prime < i} f_j(M_{i^\prime})}{\left(p_{\mathbf{1}, \mathbf{0}}(M_i)\right)^{k^{j+1}} \prod_{j^\prime < j} \left(\eta_{j^\prime}(M_i)\right)^{k^{j+1-j^\prime}}} \right)^{1/k}\]

\[\xi_j(M_i) = \lim_{k \to \infty} \left( \frac{\sum_{i^{\prime}=1}^m g_j(M_{i^\prime}) - \sum_{i^\prime < i} g_j(M_{i^\prime}) }{\left(p_{\mathbf{1}, \mathbf{0}}(M_i)\right)^{k^{n+j+1}} \left( \prod_{j^\prime = 1}^n \left(\eta_{j^\prime}(M_i)\right)^{k^{n+j+1-j^\prime}}  \right) \prod_{j^\prime < j} \left(\xi_{j^\prime}(M_i)\right)^{k^{j+1-j^\prime}}} \right)^{1/k}.\]
\end{theorem} 
\begin{proof}
We can determine the necessary values of $p_{\mathbf{1}, \mathbf{0}}(M_i)$ using theorem \ref{recovering-algebra}. The values $\sum_{i^{\prime}=1}^m f_j(M_{i^\prime})$ and $\sum_{i^{\prime}=1}^m g_j(M_{i^\prime})$ are in the ring of algebraic functions on $\calR(k,n)$. All other values of $\eta_j(M_i)$ and $\xi_j(M_i)$ in the limits shown above can be determined by induction.

The evaluation of the limits follows from lemma \ref{recovering-technical} with $w_i = 1$, $z_{i,1} = p_{\mathbf{1}, \mathbf{0}}(x_i)$, $z_{i,j+1} = \eta_j (x_i)$ for $1 \leq j \leq n$, and $z_{i,j+1} = \xi_j (x_i)$ for $n+1 \leq j \leq 2n$.
\end{proof}

\begin{corollary}
A one dimensional persistence module $M$ is completely recoverable from the values $\{p_{a,b}(M)\}_{(a,b) \in \N_+ \times \N}$ (equivalently, from the values $\{F_{a,b}(M)\}_{(a,b) \in \N_+ \times \N}$).
\end{corollary}
\begin{proof}
This follows from the fact that $\calM(k,1) = \calR(k,1)$.
\end{proof}

The previous two theorems merely give an idea of how one may use lemma \ref{recovering-technical} to recover information about an element $M \in \calR(k,n)$. One may, of course, use lemma \ref{recovering-technical} with other functions defined on $\calR(k,n)$.

\bibliographystyle{alpha}
\bibliography{multi_invariants}
\end{document}